\documentclass[12pt]{elsarticle}
\usepackage[margin=2.5cm]{geometry}% by courtesy of Mico
\usepackage{lipsum}
\usepackage{etoolbox}
\usepackage{ifthen}

\makeatletter
\def\ps@pprintTitle{%
   \let\@oddhead\@empty
   \let\@evenhead\@empty
   \def\@oddfoot{\reset@font\hfil\thepage\hfil}
   \let\@evenfoot\@oddfoot
}
\makeatother
\usepackage{amsfonts,color,morefloats,pslatex}
\usepackage{amssymb,amsthm, amsmath,latexsym}
\usepackage{amsmath,amscd}
\usepackage{enumitem}
\usepackage{xargs}

\newcommand{\mqi}[1]{{{\color{black}#1}}}
\def\qi#1 {\fbox {\footnote {\ }}\ \footnotetext { From Qi: {\color{red}#1}}}

\newtheorem{theorem}{Theorem}
\newtheorem{lemma}[theorem]{Lemma}

\newtheorem{example}[theorem]{Example}

\newcommand{\bb}{\mathbb}
\newcommand{\tr}[3][1]{\mathrm{Tr}^{#2}_{#1} \left( #3 \right)}

\newcommand{\F}[2][2]{\ifthenelse{\equal{#2}{1}}{\mathbb{F}_{#1}}{\mathbb{F}_{#1^{#2}}}}

\newcommand{\llp}{\left(}
\newcommand{\rrp}{\right)}
\newcommand{\llb}{\left\{}
\newcommand{\rrb}{\right\}}
\newcommand{\wsh}[2]{\mathcal{W}_{#1} \left( #2 \right) }
\newcommand{\efun}[2][-1]{\ifthenelse{\equal{#1}{-1}}{\left( -1 \right)^{#2}}{#1^{#2}}}
\newcommand{\wt}[2][q]{\ifthenelse{\equal{#1}{q}}{\mathrm{wt}_{q}\left( #2 \right)}{\mathrm{wt}_{#1}\left( #2 \right)}}
\newcommand{\f}{f_{a,b}}
\newcommand{\fb}{f_{a,1}}
\newcommand{\ff}[1]{ f_{a,b}\left( #1\right)}
\newcommand{\mcs}{2^m-1}
\newcommand{\fcs}{\frac{2^n-1}{3}}

\usepackage{blindtext}

\begin{document}
%\tableofcontents

\begin{frontmatter}

%% Title, authors and addresses

%% use the tnoteref command within \title for footnotes;
%% use the tnotetext command for the associated footnote;
%% use the fnref command within \author or \address for footnotes;
%% use the fntext command for the associated footnote;
%% use the corref command within \author for corresponding author footnotes;
%% use the cortext command for the associated footnote;
%% use the ead command for the email address,
%% and the form \ead[url] for the home page:
%%
%% \title{Title\tnoteref{label1}}
%% \tnotetext[label1]{}
%% \author{Name\corref{cor1}\fnref{label2}}
%% \ead{email address}
%% \ead[url]{home page}
%% \fntext[label2]{}
%% \cortext[cor1]{}
%% \address{Address\fnref{label3}}
%% \fntext[label3]{}

%\title{The linear codes of $2$-designs held in a class of ternary linear codes

\title{ An Open Problem on the Bentness of Mesnager's Functions}

\tnotetext[fn1]{
This work was supported by the National Nature Science Foundation of China under Grant No. 11871058, 11931005 and 11971321,  the National Key Research and Development
Program of China under Grant No. 2018YFA0704703, and Zhejiang provincial Natural Science Foundation of China under Grant No. LY21A010013.
%C. Ding's research was supported by the Hong Kong Research Grants Council,
%Proj. No. 16300418.
}
%}
\author[cmt]{Chunming Tang}
\ead{tangchunmingmath@163.com}

\author[cmt]{Peng Han}
\ead{penghanmitp@163.com}

%\author[cding]{Cunsheng Ding}
%\ead{cding@ust.hk}

\author[qwang]{Qi Wang}
\ead{wangqi@sustech.edu.cn}
\author[jzhang]{Jun Zhang}
\ead{junz@cnu.edu.cn}
\author[yqi]{Yanfeng Qi}
\ead{qiyanfeng07@163.com}

%\cortext[lcj]{Corresponding author}

\address[cmt]{School of Mathematics and Information, China West Normal University,  Nanchong, Sichuan,  637002,  China}

%\address[phan]{School of Mathematics and Information, China West Normal University, Nanchong, Sichuan,  637002, China}

%\address[cding]{Department of Computer Science and Engineering, The Hong Kong University of Science and Technology, Clear Water Bay, Kowloon, Hong Kong, China}

\address[qwang]{Department of Computer Science and Engineering,  Southern University of Science and Technology,  Shenzhen 518055,  China}

\address[jzhang]{School of Mathematical Sciences,  Capital Normal University,  Beijing 100048,  China}
\address[yqi]{School of Science,  Hangzhou Dianzi University,  Hangzhou,  Zhejiang,  310018,  China}

%% use optional labels to link authors explicitly to addresses:
%% \author[label1,label2]{<author name>}
%% \address[label1]{<address>}
%% \address[label2]{<address>}
%\author{Cunsheng Ding}
%\ead{cding@ust.hk}

%\cortext[lcj]{Corresponding author}
%\address{Department of Computer Science and Engineering,
%The Hong Kong University of Science and Technology,
%Clear Water Bay, Kowloon, Hong Kong, China}

%\tableofcontents

\begin{abstract}
%There are many fascinating connections between linear codes, finite groups and %combinatorial $t$-designs.

Let $n=2m$.  In the present paper,  we study \mqi{the} binomial Boolean functions of the form $$\ff{x} = \tr{n}{a x^{\mcs}} +\tr{2}{b x^{\fcs}}, $$
where $m$ is an even positive integer,  $a\in \F{n}^*$ and $b\in \F[4]{1}^*$.
We show that $\f$ is a bent function if the Kloosterman sum %$K_m \llp a^{2^m+1} \rrp$
$$K_{m}\llp a^{2^m+1} \rrp=1+ \sum_{x\in \F{m}^*} \efun{\tr{m}{a^{2^m+1} x+ \frac{1}{x}}}$$
equals $4$,  thus settling an open problem of Mesnager. \mqi{The proof employs tools including computing Walsh coefficients of Boolean functions via multiplicative
characters, divisibility properties of Gauss sums,  and graph theory.}
%The proof employs diverse methods involving computing Walsh coefficients of Boolean functions via multiplicative character,  divisibility properties of Gauss sums,  and graph theory.

\end{abstract}

\begin{keyword}
Boolean function \sep bent function \sep Walsh transform  \sep Gauss Sum
\sep directed graph.

%% PACS codes here, in the form: \PACS code \sep code

%% MSC codes here, in the form: \MSC code \sep code
%% or \MSC[2008] code \sep code (2000 is the default)
\MSC  05B05 \sep 51E10 \sep 94B15

\end{keyword}

\end{frontmatter}

%\tableofcontents

\section{Introduction}

Let $\F[q]{1}$ be the finite field of $q$ elements,  where $q = p^n$ and $p$ is a prime.
Let $f$ be a Boolean function   from
 $\F{n}$ to $\F{1}$ in univariate trace
form.
It's \emph{Walsh coefficient}  at $b\in \F{n}$  can be defined as
\begin{eqnarray}
\wsh{f}{b}=\sum_{x \in \F{n}} \efun{f(x)+\tr{n}{bx}},
\end{eqnarray}
where $\mathrm{Tr}^{n}_1 (\cdot)$ is the \emph{absolute trace function} from $\F{n}$ to $\F{1}$.
The function $f$ is said to be a bent function if its Walsh
coefficients $\wsh{f}{b}$ take values $\pm 2^m$ only for all $b\in \F{n}$.
Bent functions \cite{Dillon74,Rothaus76} are the indicators of Hadamard difference sets in elementary Abelian $2$-groups.
They play \mqi{important} roles in symmetric cryptography,  coding theory,  combinatorial designs,  sequences, \mqi{etc}.
A Boolean (bent) function $f$ over $\F{n}$ is called hyper-bent if $f(x^k)$ is bent for
every $k$ co-prime \mqi{to} $2^n-1$.
Bent functions exist \mqi{for every even integer $n$}.  We shall denote $n=2m$ in the sequel.
If $f$ is a bent function,  then there exists a Boolean function,  that we shall
denote by $\tilde{f}$,  such that,  for any $b$ in $\F{n}$,
\[\wsh{f}{b}=2^m \efun{\tilde{f}(b)}.\]
This function $\tilde{f}$ is bent too\mqi{, and we} call it the dual of $f$.

It is usually difficult to compute the Walsh coefficients of Boolean functions  explicitly;
sometimes,  even computing the absolute values of Walsh coefficients  is
difficult.  However, such difficulties can sometimes be bypassed
by divisibility considerations.
In fact,  the condition on Walsh coefficients of bent functions can be weakened,
without losing the property of being necessary and sufficient \cite{LL08}:
\begin{lemma}\label{lem:Bent-Mod}
Let $n=2m$.  A Boolean function $f$ over $\F{n}$ is bent if and only if for each $b\in \F{n}^*$,  $\wsh{f}{b} \equiv 2^{m} \mod{2^{m+1}}$.  In that case,
the dual $\tilde{f}$ of $f$ can be determined by approximations modulo
\[ \wsh{f}{b} \equiv  (-1)^{\tilde{f}(b)}2^{m}  \pmod{2^{m+2}}.\]
\end{lemma}
For more details about bent functions and their applications
 in cryptography and coding theory,  see  \cite{CarletBook} and \cite{SihemBook}.

Finding bent functions in univariate trace
form is in general difficult and \mqi{of} theoretical interest,  since it gives more
insight on bent functions.  Moreover,  the output to such functions is often faster
to compute thanks to their particular form.

In 2009,  Mesnager \cite{Mes09} has exhibited an infinite class of binomial Boolean functions defined
on $\F{n}$ whose expression is the sum of a Dillon monomial function and a trace
function of the from $\tr{2}{x^{\frac{2^n-1}{3}}}$.   More precisely, \mqi{they are of the form}
\begin{eqnarray}\label{eqn:sihem-functions}
\ff{x} = \tr{n}{a x^{\mcs}} +\tr{2}{b x^{\fcs}},
\end{eqnarray}
where $n=2m$,  $a \in \F{n}^*$ and $b\in \F[4]{1}^*$.  When $m$ is odd,  Mesnager \cite{Mes09,MesDCC} has shown that such functions are bent if and only if the Kloosterman sum
$K_{m}\llp a^{2^m+1} \rrp=1+ \sum_{x\in \F{m}^*} \efun{\tr{m}{a^{2^m+1} x+ \frac{1}{x}}}$ associated with $a^{2^m+1}$ is equal to $4$.
Not only does such a criterion \mqi{give} a concise and elegant characterization for bentness,
but using the connection between Kloosterman sums and elliptic curves it also
allows  fast generation of bent functions.
Unfortunately,  the proof of the aforementioned characterization does not extend to
the case where $m$ is even.
When $m$ is even,  the situation seems to be more complicated
than \mqi{that} in the odd case.
In particular, we are not able to say if a function $\f$  is or not in the Partial Spread class, {\color{black} which is a key part in Mesnager's proof for the odd case}.
Nevertheless,  it has been
shown that $K_{m}\llp a^{2^m+1} \rrp=4$ is still a necessary condition for $\f$ to be bent~\cite{Mes11},  but it is an open problem
to tell whether this condition is sufficient for all even $m$ or not (see  \cite[Open Problem 3]{Carlet14} and \cite{Carlet-Sihem}).  Further experimental evidence gathered by Flori,  Mesnager and
Cohen \cite{FMC} supported the conjecture that it should also be a sufficient condition: for $m$
up to $16$,  $\f$ is bent if and only if $K_{m}\llp a^{2^m+1} \rrp=4$.
In the case when \mqi{$m \equiv 2 \bmod{4}$}, Flori \cite{F16} presented a link between the bentness of $\f$ and a
conjecture about exponential sums involving Kloosterman sums.

We will show that the binomial function $\f$ is indeed bent if $K_{m}\llp a^{2^m+1} \rrp=4$,  whose
proof is the \mqi{objective} of the present paper.  We state our main
result  as \mqi{the following.}
\begin{theorem}\label{thm:Sihem}
Let $m$ be an even positive integer,  and let $n=2m$.  Let $a\in \F{n}^*$ and $b\in \F[4]{1}^*$.  If the Kloosterman sum $K_{m}\llp a^{2^m+1} \rrp=4$, then the binomial function $\f: \F{n} \rightarrow  \F{1}$ defined in Equation~(\ref{eqn:sihem-functions})
%\[\ff{x}= \tr{n}{a x^{\mcs}} +\tr{2}{b x^{\fcs}}\]
is a  bent function.
%where $K_{m}\llp a^{2^m+1} \rrp=1+ \sum_{x\in \F{m}^*} \efun{\tr{m}{a^{2^m+1} x+ \frac{1}{x}}}$.
\end{theorem}

The \mqi{remainder of the} paper is organized as follows.  In Section \ref{sec:Gauss-sums},  we fix our main notation and recall the
necessary background on character sums over finite fields.  Next, in Section \ref{sec:main},  we give the proofs that
the value $4$ of \mqi{the} Kloosterman sum \mqi{indeed leads to} bent functions.
Finally,  in Section \ref{sec:weight},  we give the detailed  discussions of binary weight inequality  needed in the proof of our main theorem.

%\section{Preliminaries}

\section{Gauss sums and Stickelberger's Theorem}\label{sec:Gauss-sums}
We collect some auxiliary results on Gauss sums as a preparation for computing \mqi{the} Walsh coefficients of Boolean functions.

Let $p$ be a prime number and $q=p^n.$ Let $\bb{Z}_p$ be the ring of integers in the field $\bb{Q}_p$ of $p$-adic rational numbers.
Let $\bb{Z}_q$ be the ring of integers in the unique unramified extension of $\bb{Q}_p$ with \mqi{the} residue field $\F[q]{1}$.
Let $\omega_q$ be the Teichm\"uller character of the multiplicative group $\F[q]{1}^*$ consisting of all nonzero elements of $\F[q]{1}$.
For $x \in \F[q]{1}^*$,  the value $\omega_q (x)$ is just the $(q-1)$-th root of unity in $\bb{Z}_q$ such that $\omega_q (x)$
modulo $p$ reduces to $x$.  For any integer $k$,  define the  Gauss sum over $\F[q]{1}$ by
\begin{eqnarray}
  G(k)=\sum_{x \in \F[q]{1}^*}  \mqi{\omega_q^{-k}(x)} \efun[\zeta_p]{\tr{n}{x}},
\end{eqnarray}
%\begin{eqnarray}
%G(k)=\sum_{x \in \F[q]{1}^*}  \omega_q(x)^{-k} \efun[\zeta_p]{\tr{n}{x}},
%\end{eqnarray}
where $\zeta_p$ is a fixed primitive $p$-th root of unity in an extension of $\bb{Q}_p$ and $\tr{n}{x}=x+x^p+\cdots+x^{p^{n-1}}$ is the trace map from $\F[q]{1}$ to $\F[p]{1}$.

Gauss sums are instrumental in the transition from the additive to the multiplicative structure of a finite field.
We have the following well-known \mqi{relation} between \mqi{the} additive characters and \mqi{the} multiplicative characters.
\begin{lemma}\label{lem:add-mult}
For all $x \in \F[q]{1}^*$,  the Gauss sums satisfy the following interpolation relation
\begin{eqnarray*}
  \efun[\zeta_p]{\tr{n}{x}}= \frac{1}{q-1} \sum_{k=0}^{q-2} G(k) \mqi{\omega_q^k(x)}.
\end{eqnarray*}
%\begin{eqnarray*}
%\efun[\zeta_p]{\tr{n}{x}}= \frac{1}{q-1} \sum_{k=0}^{q-2} G(k) \omega_q(x)^k.
%\end{eqnarray*}
\end{lemma}

Let $k$ be any integer not divisible by $q-1$.   Then there are unique integers $k_0,  k_1,  \dots,  k_{n-1}$ with $0 \le k_i \le p-1$ for all
$i$,  $0 \le i \le n-1$ such that
\[ k \equiv k_0+ k_1 p+ \dots +k_{n-1} p^{n-1} \pmod{q-1}.\]
We define the ($p$-ary) weight  of $k$ modulo $q-1$,  denoted by $\wt{k}$,  as
$\wt{k}=k_0+k_1 +\dots +k_{n-1}$.  For integers $k$ divisible by $q-1$,  we define $\wt{k}=0$.

To get the first nonzero digit in the $\pi$-adic expansion of the Gauss sum,
we can use the Stickelberger theorem (see \cite{GK79}  and  \cite{Lang}).
\begin{theorem}[Stickelberger]\label{thm:stick}
For integer $0 \le  k \le  q - 2$,  write $k =k_0 + k_1p + \dots  + k_{n-1}p^{n-1}$
in $p$-adic expansion,  where $0 \le  k_i \le p-1$.  Then,
\begin{eqnarray*}
G(k) \equiv \frac{- \pi^{ \wt{k}}}{k_0 ! k_1 ! \dots k_{n-1}!}  \pmod{\pi^{\wt{k}+p-1}},
\end{eqnarray*}
where $\pi$ is the unique element in $\bb{Z}_p[\zeta_p]$ satisfying
$\pi^{p-1} =-p $ and $\pi \equiv \zeta_p-1 \pmod{(\zeta_p-1)^2}$.
In particular,  if $p=2$
the
following congruence holds:

\begin{eqnarray*}
G(k) \equiv 2^{\wt{k}}  \pmod{2^{\wt{k}+1}}.
\end{eqnarray*}
\end{theorem}

\section{Proof of the main results}\label{sec:main}
The goal of this section is to prove Theorem \ref{thm:Sihem}.  In order to do so, we must first  establish
 some results  concerning exponential sums over finite fields.

\begin{lemma}\label{lem:kloosterman-gauss}
Let $n=2m$,  $q=2^n$ and $a\in \F{m}^*$.  It holds that
\begin{eqnarray*}
\sum_{x\in \F{m}^*} \efun{ \tr{m}{ax +\frac{1}{x}}}= \frac{1}{2^m-1}\sum_{0\le i \le 2^m-2}    \overline{G}(i)^2 \omega_q^{(2^m+1)i} \llp  a \rrp,
\end{eqnarray*}
where $\overline{G}(i)$ denotes the Gauss sum  $\sum_{x \in \F{m}^*}  \omega_q^{-(2^m+1)i} \llp {x} \rrp \efun{\tr{m}{x}}$
over $\F{m}$ with respect to the multiplicative character $ \omega_q^{(2^m+1)i}$.
\end{lemma}
\begin{proof}
Note that $\omega_q^{2^m+1}$ is a generator of the multiplicative character group of $\F{m}^*$.
By Lemma \ref{lem:add-mult},  we have
\begin{eqnarray*}
&&\sum_{x\in \F{m}^*} \efun{ \tr{m}{ax +\frac{1}{x}}}\\
 &&=\frac{1}{(2^m-1)^2}\sum_{0\le i,j \le 2^m-2}    \overline{G}(i) \overline{G}(j) \sum_{x\in \F{m}^*} \omega_q^{(2^m+1)i} \llp  ax \rrp   \omega_q^{-(2^m+1)j} \llp  x \rrp \\
 &&=\frac{1}{(2^m-1)^2}\sum_{0\le i,j \le 2^m-2}    \overline{G}(i) \overline{G}(j)  \omega_q^{(2^m+1)i} \llp  a \rrp    \sum_{x\in \F{m}^*} \omega_q^{(2^m+1)(i-j)} \llp  x \rrp .
\end{eqnarray*}
Since $\sum_{x\in \F{m}^*} \omega_q^{(2^m+1)(i-j)} \llp  x \rrp=2^m-1$ if and only if $i\equiv j \pmod{2^m-1}$, \mqi{and is $0$ otherwise,}  we thus have
\begin{eqnarray*}
\sum_{x\in \F{m}^*} \efun{ \tr{m}{ax +\frac{1}{x}}}=\frac{1}{2^m-1}\sum_{0\le i \le 2^m-2}    \overline{G}(i)^2 \omega_q^{(2^m+1)i} \llp  a \rrp,
\end{eqnarray*}
which completes the proof.
\end{proof}

\begin{lemma}\label{lem:dillon-conguence}
Let $n=2m$ and $q=2^n$.  Write $B=\sum_{x \in \F{n}^*} \efun{\tr{n}{a x^{\mcs } + c x} }$,  where $a, c \in \F{n}^*$.
Then
\begin{enumerate}[label=(\arabic*)]
\item $B= \frac{1}{q-1}\sum_{0\le i \le q-2}    G(i) G\llp -\llp \mcs \rrp i \rrp \omega_q \llp a^i c^{-\llp \mcs \rrp i } \rrp$.
\item $B\equiv K_m \llp a^{2^m+1} \rrp+2^m-1 \pmod{2^{m+1}}$. % where $K_m \llp a^{2^m+1} \rrp=1+\sum_{x\in \F{m}^*} \efun{\tr{m}{a^{2^m+1}x+\frac{1}{x}}}$.
\end{enumerate}
\end{lemma}

\begin{proof}
By Lemma \ref{lem:add-mult},  we have
\begin{eqnarray*}
&& \sum_{x \in \F{n}^*} \efun{\tr{n}{a x^{\mcs } + c x} }\\
&&=\frac{1}{(q-1)^2}\sum_{0\le i,j \le q-2}    G(i) G(j) \sum_{x\in \F{n}^*} \mqi{\omega_q^i \llp a x^{\mcs } \rrp}    \mqi{\omega_q^j \llp  c x \rrp} \\
&&=\frac{1}{(q-1)^2}\sum_{0\le i,j \le q-2}    G(i) G(j) \omega_q \llp a^i c^j \rrp \sum_{x\in \F{n}^*} \mqi{\omega_q^{ \llp \mcs \rrp  i + j} \llp  x \rrp } \\
 &&= \frac{1}{(q-1)^2}\sum_{0\le i,j \le q-2}    G(i) G(j) \omega_q \llp a^i c^j \rrp  \omega_q ^{ \llp \mcs \rrp  i + j} \llp  \F{n}^* \rrp ,
\end{eqnarray*}
%\begin{eqnarray*}
%&& \sum_{x \in \F{n}^*} \efun{\tr{n}{a x^{\mcs } + c x} }\\
 %&&=\frac{1}{(q-1)^2}\sum_{0\le i,j \le q-2}    G(i) G(j) \sum_{x\in \F{n}^*} \omega_q \llp a x^{\mcs } \rrp ^i   \omega_q \llp  c x \rrp ^j\\
 %&&=\frac{1}{(q-1)^2}\sum_{0\le i,j \le q-2}    G(i) G(j) \omega_q \llp a^i c^j \rrp \sum_{x\in \F{n}^*} \omega_q \llp  x \rrp ^{ \llp \mcs \rrp  i + j}\\
 %&&= \frac{1}{(q-1)^2}\sum_{0\le i,j \le q-2}    G(i) G(j) \omega_q \llp a^i c^j \rrp  \omega_q ^{ \llp \mcs \rrp  i + j} \llp  \F{n}^* \rrp ,
%\end{eqnarray*}
where $ \omega_q ^{ \llp \mcs \rrp  i + j} \llp  \F{n}^* \rrp=  \sum_{x\in \F{n}^*} \omega_q \llp  x \rrp ^{ \llp \mcs \rrp  i + j} $.
It is well known that
\begin{eqnarray*}
\omega_q ^{ \llp \mcs \rrp  i + j} \llp  \F{n}^* \rrp=\begin{cases}
q-1,   &  if ~~  j \equiv  -\llp \mcs \rrp  i  \pmod{q-1} ,\\
0,    & otherwise.
\end{cases}
\end{eqnarray*}
We thus have
\begin{eqnarray*}
\begin{array}{rl}
& \sum_{x \in \F{n}^*} \efun{\tr{n}{a x^{ \mcs } + c x} } \\
 &=\frac{1}{q-1}\sum_{0\le i \le q-2}    G(i) G\llp -\llp \mcs \rrp i \rrp \omega_q \llp a^i c^{-\llp \mcs \rrp i } \rrp.
\end{array}
\end{eqnarray*}
This completes the proof of the first assertion of the lemma.

To prove part (2),   recall that $\wt{-\llp \mcs \rrp i }=m$ if $i$ is not a multiple of $2^m+1$ (see  \cite[Lemma 2]{LL08}).
Combining Theorem \ref{thm:stick} with the fact $G(0)=-1$ yields
\begin{eqnarray}\label{eqn:dillon-mod}
\begin{array}{rl}
& \sum_{x \in \F{n}^*} \efun{\tr{n}{a x^{ \mcs } + c x} }\\
 &  \equiv  \sum_{0\le i \le 2^m-2}    G((2^m+1)i )  \omega_q \llp a^{(2^m+1)i}  \rrp \pmod{2^{m+1}}.
\end{array}
\end{eqnarray}
Observe that $\omega_q ^{ -(2^m+1)i}\llp x \rrp = \omega_q^{-(2^m+1)i} \llp \sqrt{x}^{2^m+1} \rrp$ for each $x \in \F{n}^*$ and $0\le i \le 2^m-2$.
We deduce that
\begin{eqnarray*}
\begin{array}{rcl}
G((2^m+1)i ) & =& \sum_{x \in \F{n}^*} \omega_q^{-(2^m+1)i}(x) \efun{\tr{n}{x}}\\
&=& \sum_{x \in \F{n}^*}  \omega_q^{-(2^m+1)i} \llp \sqrt{x}^{2^m+1} \rrp \efun{\tr{n}{x}}\\
&=& \sum_{x \in \F{n}^*}  \omega_q^{-(2^m+1)i} \llp \sqrt{x}^{2^m+1} \rrp \efun{\tr{n}{\sqrt{x}}}\\
&=& \sum_{x \in \F{n}^*}  \omega_q^{-(2^m+1)i} \llp {x}^{2^m+1} \rrp \efun{\tr{n}{x}}\\
&=& - \llp  \sum_{x \in \F{m}^*}  \omega_q^{-(2^m+1)i} \llp {x} \rrp \efun{\tr{m}{x}} \rrp ^2,
\end{array}
\end{eqnarray*}
where the last equality follows by the Davenport and Hasse’s lifting theorem (see  Weil \cite[p.  503--505]{Weil49}).
Substituting $G((2^m+1)i ) $ in the expression (\ref{eqn:dillon-mod}) by $-\overline{G}(i)^2$,  we have
\begin{eqnarray*}
\begin{array}{rl}
& \sum_{x \in \F{n}^*} \efun{\tr{n}{a x^{ \mcs } + c x} }\\
 &  \equiv - \sum_{0\le i \le 2^m-2}    \overline{G}(i)^2 \omega_q \llp a^{(2^m+1)i}  \rrp \pmod{2^{m+1}}.
\end{array}
\end{eqnarray*}
Since $\omega_q \llp a^{(2^m+1)i}  \rrp= \omega_q^{(2^m+1)i} \llp \sqrt{a}^{2^m+1}  \rrp $,  we have
\begin{eqnarray*}
\begin{array}{rl}
& \sum_{x \in \F{n}^*} \efun{\tr{n}{a x^{ \mcs } + c x} }\\
 &  \equiv - \sum_{0\le i \le 2^m-2}    \overline{G}(i)^2 \omega_q^{(2^m+1)i} \llp \sqrt{a}^{2^m+1}  \rrp\\
 & \equiv  -(2^m-1)\cdot \sum_{x \in \F{m}^*}  \efun{ \tr{m}{\sqrt{a}^{2^m+1} x +\frac{1}{x}}} \pmod{2^{m+1}},
\end{array}
\end{eqnarray*}
where the last congruence follows from Lemma \ref{lem:kloosterman-gauss}.
Together with the fact $$\sum_{x \in \F{m}^*}  \efun{ \tr{m}{\sqrt{a}^{2^m+1} x +\frac{1}{x}}}
= \sum_{x \in \F{m}^*}  \efun{ \tr{m}{{a}^{2^m+1} x +\frac{1}{x}}},$$  we have
\begin{eqnarray*}
\begin{array}{c}
 \sum_{x \in \F{n}^*} \efun{\tr{n}{a x^{ \mcs } + c x} } \equiv  (2^m+1)\cdot \llp K_m \llp a^{2^m+1} \rrp-1 \rrp \pmod{2^{m+1}},
\end{array}
\end{eqnarray*}
Since $K_m \llp a^{2^m+1} \rrp$ is an even integer,  we have
\begin{eqnarray*}
\begin{array}{c}
 \sum_{x \in \F{n}^*} \efun{\tr{n}{a x^{ \mcs } + c x} } \equiv  K_m \llp a^{2^m+1} \rrp+2^m-1 \pmod{2^{m+1}},
\end{array}
\end{eqnarray*}
which completes the proof of Part (2).
\end{proof}

The most important ingredient in the proof of Theorems \ref{thm:Sihem} is the following result.
\begin{theorem}\label{thm:string}
Let $q=2^{2m}$,  where $m$ is any positive integer,  and let $u$ equal $\frac{2^{2m}-1}{3} \text{ or } 2 \cdot  \frac{2^{2m}-1}{3}$.  For any integers $a$ and $b$,  if $s$ and $t$
satisfy
\[ s\equiv u-a+b,  \quad t\equiv u+a-b \pmod{2^{2m}-1} , \]
then $\wt{a} +\wt{b}+\wt{s}+\wt{t}\ge 2m$.
\end{theorem}

The proof of Theorem \ref{thm:string} involves graph-theory methods. To streamline the presentation of the paper,
we \mqi{postpone} the proof to the next section.

The following lemma establishes \mqi{a congruence relation between} the Walsh coefficients of $\fb$ and \mqi{the} Kloosterman sum,  which plays an important role
in the proof of Theorem \ref{thm:Sihem}.
\begin{lemma}\label{lem:Walsh-Maim}
Let $a \in \F{n}^*$,  where $n=2m$ and $m$ is an even positive integer.   Let $\fb$ be the binomial function given in (\ref{eqn:sihem-functions}).
Then for any $c\in \F{n}^*$ we have
\[\wsh{\fb}{c} \equiv \frac{4-K_{m} \llp a^{2^m+1} \rrp}{3} +2^m   \pmod{2^{m+1}}.\]
\end{lemma}

\begin{proof}

Let $C_0=\llb x^3: x\in \F{n}^* \rrb$.  By noting that
$$
\tr{2}{ x^{\fcs}}=\begin{cases}
0,   &  if ~~ x\in C_0 \cup \llb 0 \rrb\\
1,  & otherwise
\end{cases},
$$
we have
\begin{eqnarray}\label{eqn:trinomial}
\begin{array}{rl}
&\wsh{\fb}{c}\\
&=\sum_{x \in \F{n}} \efun{\tr{n}{a x^{\mcs} + c x} +\tr{2}{ x^{\fcs}}}\\
&=1- \sum_{x \in \F{n}^*} \efun{\tr{n}{a x^{\mcs} + c x} } +2 \cdot \sum_{x \in C_0} \efun{\tr{n}{a x^{\mcs} + c x} }\\
&=1- \sum_{x \in \F{n}^*} \efun{\tr{n}{a x^{\mcs} + c x} } +\frac{2}{3} \cdot \sum_{x \in \F{n}^*} \efun{\tr{n}{a x^{3\llp \mcs \rrp} + c x^3} }.
\end{array}
\end{eqnarray}
By Lemma \ref{lem:add-mult},  we have
\begin{eqnarray*}
&& \sum_{x \in \F{n}^*} \efun{\tr{n}{a x^{3\llp \mcs \rrp} + c x^3} }\\
 &&=\frac{1}{(q-1)^2}\sum_{0\le i,j \le q-2}    G(i) G(j) \sum_{x\in \F{n}^*} \omega_q \llp a x^{3\llp \mcs \rrp} \rrp ^i   \omega_q \llp  c x^3 \rrp ^j\\
 &&=\frac{1}{(q-1)^2}\sum_{0\le i,j \le q-2}    G(i) G(j) \omega_q \llp a^i c^j \rrp \sum_{x\in \F{n}^*} \omega_q \llp  x \rrp ^{3\cdot \llp \mcs \rrp i + 3j}.
\end{eqnarray*}
Note that %$\sum_{x\in \F{n}^*} \omega_q \llp  x \rrp ^{3\cdot \llp \mcs \rrp i + 3j}$ equals $q-1$ or $0$according as $3\cdot \llp \mcs \rrp i + 3j \equiv 0 \pmod{q-1}$ or $3\cdot \llp \mcs \rrp i + 3j \not\equiv 0 \pmod{q-1}$.
\begin{align*}
  \sum_{x\in \F{n}^*} \omega_q \llp  x \rrp ^{3\cdot \llp \mcs \rrp i + 3j}=\begin{cases}
   q-1, & \mbox{if $3\cdot \llp \mcs \rrp i + 3j \equiv 0 \pmod{q-1}$,}\\
  0, & \mbox{otherwise.}
  \end{cases}
\end{align*}
Since $3\cdot \llp \mcs \rrp i + 3j \equiv 0 \pmod{q-1}$
if and only if $j= \frac{2^n-1}{3} \ell -(2^m-1)i$,  where $0 \le \ell \le 2$,  we obtain
\begin{eqnarray}\label{eqn:b33}
\begin{array}{rl}
& \sum_{x \in \F{n}^*} \efun{\tr{n}{a x^{3\llp \mcs \rrp} + c x^3} }\\
 &=\frac{1}{q-1}\sum_{0\le i \le q-2}    G(i) G\llp \fcs-\llp \mcs \rrp i \rrp \omega_q \llp a^i c^{\fcs-\llp \mcs \rrp i } \rrp \\
 &\quad    + \frac{1}{q-1}\sum_{0\le i \le q-2}    G(i) G\llp 2 \cdot \fcs-\llp \mcs \rrp i \rrp \omega_q \llp a^i c^{2\cdot  \fcs-\llp \mcs \rrp i } \rrp \\
 &\quad    + \frac{1}{q-1}\sum_{0\le i \le q-2}    G(i) G\llp -\llp \mcs \rrp i \rrp \omega_q \llp a^i c^{-\llp \mcs \rrp i } \rrp .
\end{array}
\end{eqnarray}
Combining (\ref{eqn:trinomial}) \mqi{with} (\ref{eqn:b33}) yields
\begin{eqnarray}\label{eqn:wsh-3part}
\begin{array}{rl}
&\wsh{\fb}{c}\\
&=1- \frac{1}{3(q-1)} \sum_{0\le i \le q-2}    G(i) G\llp -\llp \mcs \rrp i \rrp \omega_q \llp a^i c^{-\llp \mcs \rrp i } \rrp \\
& \quad + \frac{2}{3(q-1)}\sum_{0\le i \le q-2}    G(i) G\llp \fcs-\llp \mcs \rrp i \rrp \omega_q \llp a^i c^{\fcs-\llp \mcs \rrp i } \rrp \\
 &\quad    + \frac{2}{3(q-1)}\sum_{0\le i \le q-2}    G(i) G\llp 2 \cdot \fcs-\llp \mcs \rrp i \rrp \omega_q \llp a^i c^{2\cdot  \fcs-\llp \mcs \rrp i } \rrp .
\end{array}
\end{eqnarray}
When $m$ is even,  it is routine to check that
\[ \fcs-\llp \mcs \rrp i  \equiv  2^m\llp \fcs+\llp \mcs \rrp i \rrp  \pmod{2^n-1}\]
and
\[ 2\cdot \fcs-\llp \mcs \rrp i  \equiv  2^m\llp 2\cdot \fcs+\llp \mcs \rrp i \rrp  \pmod{2^n-1}.\]
\mqi{Applying} Theorem \ref{thm:string} with $u=\fcs$,  $a=i$ and $b=2^m \cdot i$\mqi{,
we then} have
\[ \wt{i}+\wt{\fcs-\llp \mcs \rrp i } \ge m. \]
Similarly,  we have
\[ \wt{i}+\wt{2\cdot \fcs-\llp \mcs \rrp i } \ge m. \]
By Theorem \ref{thm:stick},  we have
\[G(i) G\llp \fcs-\llp \mcs \rrp i \rrp  \equiv 0 \pmod{2^m}\]
and
\[G(i) G\llp 2\cdot \fcs-\llp \mcs \rrp i \rrp  \equiv 0 \pmod{2^m}.\]
Continuing from (\ref{eqn:wsh-3part}),  we have
\begin{eqnarray*}
\begin{array}{c}
\wsh{\fb}{c} \equiv 1+ \frac{1}{3} \sum_{0\le i \le q-2}    G(i) G\llp -\llp \mcs \rrp i \rrp \omega_q \llp a^i c^{-\llp \mcs \rrp i } \rrp \pmod{2^{m+1}}.
\end{array}
\end{eqnarray*}
By Lemma \ref{lem:dillon-conguence},  we get
\begin{eqnarray*}
\begin{array}{c}
\wsh{\fb}{c} \equiv 1+ \frac{1-K_m\llp a^{2^m+1} \rrp +2^m }{3}  \pmod{2^{m+1}},
\end{array}
\end{eqnarray*}
\mqi{which further gives}
%A direct computation shows that
\begin{eqnarray*}
\begin{array}{c}
\wsh{\fb}{c} \equiv \frac{4-K_m\llp a^{2^m+1} \rrp }{3}  +2^m  \pmod{2^{m+1}}.
\end{array}
\end{eqnarray*}
This completes the proof of the lemma.
\end{proof}

Having dealt with these preliminaries, we can now prove the \mqi{main} theorem.

%\noindent \textbf{Poof of Theorem \ref{thm:Sihem}.}

\begin{proof}[Proof of Theorem \ref{thm:Sihem}]
Substituting $K_m\llp a^{2^m+1} \rrp$ in the expression for $\wsh{\fb}{c}$ in Lemma \ref{lem:Walsh-Maim} by $4$,  we have
\[\wsh{\fb}{c} \equiv 2^m \pmod{2^{m+1}}.\]
Combining the above congruence with Lemma \ref{lem:Bent-Mod},  the proof of Theorem \ref{thm:Sihem}  is concluded,
by noting that $\f$ is bent if and only if $\fb$ is bent \cite{FM}.
\end{proof}

\begin{example}
Let $\F{6}$ be the finite field represented as  $\F{1}[z]/(z^6 + z^4 + z^3 + z + 1)$
and $a=z^3$. Then $K_6(a)=4$ and the binomial function $\fb (x)=\tr{12}{a x^{2^6-1}}
+\tr{2}{x^{\frac{2^{12}-1}{3}}}$ over $\F{12}$ is bent. Its dual $\tilde{f}_{a,1}(x)$
is given by $\tr{12}{z^{48} x^{357}}+
\tr{12}{z^{28} x^{147}}+
\tr{12}{z^{3} x^{63}}+
\tr{12}{z^{62} x^{21}}+
\tr{12}{z^{60} x^{105}}+
\tr{4}{ x^{273}}+ \tr{2}{x^{1365}}. $
However, the Boolean function $\fb(x^{11})$
is not bent as  its Walsh spectrum $\left\{ \wsh{\fb}{c}: c\in \F{12} \right\}$ is
$ \left\{  0, \pm 32, \pm 64, \pm 96, \pm 128, \pm 160  \right\}.$
Note that $\mathrm{g.c.d}(11,2^{12}-1)=1$. Hence $\fb$ is not a hyper-bent funcition.
\end{example}

%\begin{theorem}
%Let $q=2^m$ and $n=2m$ with $m$ even.   Let $a\in \F[q]{2}^*$.  Then the monomial function $\tr{n}{a x^{{(q-1)/3}}}$
%is bent if and only if  the monomial function  $\tr{n}{a x^{q-1}}$ with Dillon exponent is bent.
%\end{theorem}

\section{Proof of the weight inequality}\label{sec:weight}
Let $\llb s_j \rrb_{j\in \mathbb Z}$  be a sequence.
If for a positive integer $n$ the terms of $\llb s_j \rrb_{j\in \mathbb Z}$ satisfy $s_j = s_{j+n}$ for all $j\in \mathbb Z$,
 then we say that $\llb s_j \rrb_{j\in \mathbb Z}$ is $n$-periodic.  If $s_j \in \{0, 1\}$ for all $j\in \mathbb Z$,  $\llb s_j \rrb_{j\in \mathbb Z}$ is said to be a binary
 sequence.   Let $0^{ \mathbb Z}$ (resp.,  $1^{ \mathbb Z}$ ) denote the sequence $\llb s_j \rrb_{j\in \mathbb Z}$ with
 $s_j =0$ (resp., $s_j=1$) for all $j$.
 Any sequence of length $n$ can be extended to a periodic sequence with period $n$.

Given $a$ and $b$,  we use a modular add-and-carry method inspired
by \cite{H-Xiang} to help compute the weights of $u-a+b$ and $u+a-b$ that appear in the inequality
in Theorem \ref{thm:string}.  The basic result we need is a technical result related to
\cite[Theorem 13]{H-Xiang}.
\begin{theorem}\label{thm:xiang}
Let $\llb s_j \rrb_{j\in \mathbb Z}$ and $\llb a^{(i)}_j \rrb_{j\in \mathbb Z}$,  $1 \le i \le r$,  be
binary sequences of period $n$ with $\llb a^{(i)}_j \rrb_{j\in \mathbb Z} \not \in \llb  0^{ \mathbb Z},  1^{ \mathbb Z} \rrb$ for some $i$.   Let $t_1, \dots,  t_r$ be nonzero integers.
Suppose that
\[s \equiv t_1 a^{(1)}+ \cdots +t_{r} a^{(r)} \pmod{2^n-1},\]
where $s=\sum_{i=0}^{n-1} s_j 2^j$ and $a^{(i)}=\sum_{j=0}^{n-1} a^{(i)}_j 2^j$.
Then there exists a unique $n$-periodic
sequence $\llb c_i \rrb_{i\in Z}$ with terms in $\llb t_{-},  t_{-}+1,  \dots, t_{+}-1 \rrb$ such that
\[ 2c_j+ s_j  = t_1 a^{(1)}_j+ \cdots +t_{r} a^{(r)}_j +c_{j-1} \quad (j \in \mathbb Z),   \]
where $t_{-}=\sum_{i , t_i<0} t_i$ and $t_{+}=\sum_{i , t_i>0} t_i$.  Moreover,  we have that
\[\sum_{j=0}^{n-1} c_j = \sum_{i=1}^{r} t_i \sum_{j=0}^{n-1} a^{(i)}_j -\sum_{j=0}^{n-1} s_j.\]
\end{theorem}

The numbers $s_j$ and $c_j$ will be referred to as the digits and carries,  respectively,
for the computation modulo $2^n- 1$ of the number $s$.We now give the promised proof of Theorem \ref{thm:string}.

%\noindent \textbf{Proof of Theorem \ref{thm:string}.}
\begin{proof}[Proof of Theorem \ref{thm:string}]

Let $s$ and $t$ be defined as in the
theorem, and write $n=2m$.  Assume that $a$,  $b$,  $u$,  $s$ and $t$
have \mqi{the} binary digits $a_j$,  $b_j$,  $u_j$,  $s_j$ and $t_j$,  for $j=0, \dots,  n-1$.
Note that
\begin{eqnarray*}
s& \equiv &u-a  + b \pmod{2^n-1},\\
t& \equiv &u+a-b  \pmod{2^n-1}.
\end{eqnarray*}
Now apply
Theorem \ref{thm:xiang} to the defining additions for $s$ and $t$.  In both
cases,  $t_{+}= 1$ and $t_{-}=-1$,  hence there are carry sequences
$\llb c_j \rrb_{j \in \mathbb Z}$ and $\llb d_j \rrb_{j \in \mathbb Z}$
with $-1 \le c_j,  d_j \le 1$ such that
\begin{eqnarray}\label{eq:sj-tj}
\begin{array}{rcl}
2 c_j + s_j & = &u_j-a_j + b_j+c_{j-1},\\
2 d_j + t_j & = &u_j+a_j-b_j+d_{j-1}.
\end{array}
\end{eqnarray}
Moreover,  $\sum_{j=0}^{n-1} c_j + \sum_{j=0}^{n-1} d_j +  \sum_{j=0}^{n-1} s_j  + \sum_{j=0}^{n-1} t_j =n$.
Using this relation, we see that the weight inequality in
the theorem is equivalent to
\begin{eqnarray}\label{eq:walk-weight}
\sum_{j=0}^{n-1} \left(  a_j +b_j - c_j -  d_j \right) \ge 0 .
\end{eqnarray}

In order to analyze the contribution of the individual binary
digits $a_j, b_j,  c_j,  d_j$ to the sum $\sum_{j=0}^{n-1} \left(  a_j +b_j - c_j -  d_j \right)$ ,
we construct the following weighted directed graph  $G$.

The graph $G$ will have a vertex $(u',  a', b',  c',  d')$ whenever
$u',  a',  b' \in \left\lbrace 0, 1 \right\rbrace$ and $c',  d' \in \left\lbrace -1,  0, 1 \right\rbrace$,
and a weighted directed arc
\[(u',  a', b',  c',  d') \xrightarrow{a'+b'-c''- d''}  (u'',  a'', b'',  c'',  d''),\]
whenever $u''=1-u'$,
\[s'= u'-a'+b'+c'-2c'' \in  \left\lbrace 0, 1 \right\rbrace, \]
and
\[t'= u'+a'-b'+d'-2d'' \in  \left\lbrace 0, 1 \right\rbrace.\]
Note that,  according to these definitions,  whenever (\ref{eq:sj-tj}) holds
there is an arc
\[(u_{j},  a_j, b_j,  c_{j-1},  d_{j-1}) \xrightarrow{a_{j}+b_{j}-c_{j}- d_{j}}  (u_{j+1},  a_{j+1}, b_{j+1},  c_{j},  d_{j}),\]
in the graph.  Therefore,  there is a one-to-one correspondence between $n$-periodic sequences $\llb u_j \rrb_{j \in \mathbb Z}$,  where $u_j+u_{j-1}=1$
,  $\llb a_j \rrb_{j \in \mathbb Z}$,  $\llb b_j \rrb_{j \in \mathbb Z}$,  $\llb c_j \rrb_{j \in \mathbb Z}$,  $\llb d_j \rrb_{j \in \mathbb Z}$
\mqi{satisfy
relation} (\ref{eq:sj-tj}) with the corresponding sum of weights
$w= \sum_{j=0}^{n-1} \left(  a_j +b_j - c_j -  d_j \right)$ and the directed walks of
length $n$ in the graph for which the sum of the weights of the
arcs equals $w$.  Thus to verify (\ref{eq:walk-weight}),  it suffices to show that
the graph does not contain any walk of strictly
negative weight.

We investigated the  weighted directed graph $G$ with the aid of a computer.
It turns out that $G$,  a digraph on $72$ vertices,  has $33$ strongly
connected components.  Here,  two vertices of a directed  graph are said to be strongly
connected if they are contained together in a directed cycle. The relation of
being strongly connected is an equivalence relation on the set of vertices; the
equivalence classes of this relation are called the the strongly connected
components of the directed graph.

One of these strongly connected components has size $40$ (that is,  contains
$40$ vertices),  denoted by $H$,  and further $32$ strongly connected components have size $1$.
Obviously,  all strongly connected components of size $1$ contain no arc at all.
 The weighted directed graph $H$ has $8$ arcs of weight $-1$,  $32$ arcs of weight $0$,
$80$ arcs of weight $1$,  $32$ arcs of weight $2$,  and further $8$ arcs of weight $3$.

So we are done if we can show that the weight of each directed
walk in the directed graph $H$ is nonnegative.
Using standard graph-theory tools of MAGMA \cite{magma} it is easy to verify that the component $H$ has no negative-weight walk.
Note that the running time is
negligible.  This completes the proof.
\end{proof}

\section{Concluding remarks}\label{sec:conc}

Combining our results with  the work of Mesnager in \cite{Mes09,MesDCC},
the bentness  of binomial function $\f(x)= \tr{n}{a x^{\mcs}} +\tr{2}{b x^{\fcs}}$,
where $m$ is an positive integer,  $a\in \F{n}^*$ and $b\in \F[4]{1}^*$,  has been completely characterized:
$\f$ is bent if and only if the Kloosterman sum $K_m(a^{2^m+1})=1+\sum_{x \in \F{n}^*} \efun{\tr{m}{a^{2^m+1}x+\frac{1}{x}}}$
equals $4$.
A related question is whether the bent function arising from Theorem \ref{thm:Sihem} is hyper-bent.  We have checked \mqi{by computer} that the answer is no in all cases when $m=4,  6,  8$ and $10$,  and we believe that the answer is no in general.

%\section*{Acknowledgments}

%The authors are grateful to the reviewers and the Editors for their comments and suggestions that improved the presentation of this paper.

%\section*{Acknowledgments}

%C. Ding's research was supported by the Hong Kong Research Grants Council,
%Proj. No. 16300418. C. Tang was supported by National Natural Science Foundation of China (Grant No.
%11871058) and China West Normal University (14E013, CXTD2014-4 and the Meritocracy Research
%Funds)

\end{document}